\documentclass[11pt,reqno]{amsart}
\usepackage{amsthm,amssymb,amsmath,calc,eucal,ifthen}
\usepackage{amscd}
\usepackage[all,arc,curve,color,frame,matrix,arrow]{xy}

\numberwithin{equation}{section}

\newtheorem{theorem}{Theorem}[section]
\newtheorem{lemma}[theorem]{Lemma}
\newtheorem{proposition}[theorem]{Proposition}
\newtheorem{corollary}[theorem]{Corollary}

\theoremstyle{remark}
\newtheorem{remark}[theorem]{Remark}
\newtheorem{example}[theorem]{Example}

\newtheoremstyle{rmdefinition}{}{}{\upshape}{}{\bfseries}{.}{ }{}

\theoremstyle{rmdefinition}
\newtheorem{definition}[theorem]{Definition}

\newcommand{\be}[1]{\begin{equation}\label{#1}}
\newcommand{\ee}{\end{equation}}
\newcommand{\beqa}{\begin{eqnarray}}
\newcommand{\eeqa}{\end{eqnarray}}

\newcounter{tmpc}
\newlength{\tmplenght}
\newlength{\tmplenghta}
\newlength{\tmplenghtb}
\newlength{\tmplenghtc}

\begin{document}

\title[Infinite volume limits]{Infinite volume limits in Euclidean quantum field theory via stereographic projection}

\author{Svetoslav Zahariev}
\maketitle
\vspace{-7pt}
\begin{center}\small
LaGuardia Community College of The City University of New York,\\
MEC Department, 31-10 Thomson Ave.\\
Long Island City, NY 11101, U.S.A.\\szahariev@lagcc.cuny.edu
\end{center}
\vspace{-5pt}
\begin{abstract} We present a general infinite volume limit construction of probability measures obeying the Glimm-Jaffe axioms of Euclidean quantum field theory in arbitrary space-time dimension. In particular, we obtain measures that may be interpreted as corresponding to scalar quantum fields with arbitrary bounded continuous self-interaction. It remains however an open problem whether this general construction contains non-Gaussian measures.
\end{abstract}
\let\thefootnote\relax\footnotetext{Research partially supported by CUNY CCR Grant No. 1510.}
\vspace{-7pt}
\tableofcontents
\newpage
\section{Introduction}
Establishing the existence of interacting field models satisfying the Wightman axioms \cite{SW} of relativistic quantum field theory (QFT) in four dimensional space-time has remained an important open problem in mathematical physics since the 1960s.
Many models in two and three space-time dimensions \cite{Su} have been constructed so far employing the framework of Euclidean QFT in which  Minkowski space is replaced with Euclidean space. In the mid 1970s Osterwalder and Schrader \cite{OS} discovered a set of axioms formulated on Euclidean space that is equivalent to the axioms of Wightman.

In this paper we present a general construction of  probability measures satisfying the Glimm-Jaffe axioms of Euclidean QFT as stated in \cite[Chapter 6]{GJ} (except possibly ergodicity) in arbitrary space-time dimension. More explicitly, we obtain  Euclidean invariant reflection-positive measures on the space of distributions on $\mathbb{R}^{d}$ possessing characteristic functionals that satisfy appropriate analyticity and regularity conditions.

Our method may be summarized as follows. Given a sequence of mollifiers on the standard $d$-dimensional sphere $\mathbb{S}^{d}$ and a sequence of densities defined on the space $\mathcal{D}(\mathbb{S}^{d})$ of smooth functions on $\mathbb{S}^{d}$ that satisfy simple integral bounds, we construct a sequence of measures  on $\mathcal{D}(\mathbb{S}^{d})$. We observe that one can obtain such a sequence of densities from an arbitrary bounded real continuous function (representing a self-interaction). We then transfer these measures to $\mathbb{R}^{d}$ utilizing the stereographic projection, via a scaling limit procedure analogous to that developed in \cite[Section 3]{Sc2}, and show that the transferred sequence contains a subsequence weakly convergent on distributions. Heuristically, this scaling limit may be envisaged as a process in which the radius of the sphere tends to infinity, the subgroup of the isometry group of  $\mathbb{S}^{d}$ preserving a given point becomes the rotation group of $\mathbb{R}^{d}$, and the remaining rotations of $\mathbb{S}^{d}$  are identified with translations in $\mathbb{R}^{d}$.

While one would expect that in many special cases the limit measures so obtained coincide with the free scalar field measure, we believe that our construction is sufficiently general to possibly contain non-Gaussian examples as well, including in four space-time dimensions. We also conjecture that feeding the scaling limit construction with the densities corresponding to the $P(\phi)_{2}$ model on  $\mathbb{S}^{2}$ (cf. \cite{BJM}) produces the well-known $P(\phi)_{2}$ model on $\mathbb{R}^{2}$.

A distinct feature of our method is that the ultraviolet and the infrared (infinite volume) limits  are obtained simultaneously, unlike in most of the standard constructions of low-dimensional Euclidean QFT models. It also appears that the method, due do its generality, may be applied to many non-scalar QFT interactions. In particular, it may be interesting to consider the case of quantum Yang-Mills theory as it is widely believed (cf. \cite{JW}) that one of  most formidable difficulties in establishing the existence of this model in $d=4$ is the passage to infinite volume.

The paper is organized as follows. In Section \ref{refposlimea} we discuss a notion of reflection positivity on Riemannian manifolds and describe a sufficient condition for the reflection positivity of the limit of a sequence of measures that are approximately reflection positive in a suitable sense. Section \ref{finitevolmes} is dedicated to the construction of limits of isometry invariant measures on a closed Riemannian manifold. Sequences of measures are produced from a given sequence of Gaussian measures on smooth functions and a bounded continuous interaction function; the existence of a (weak) limit point on distributions is established as a special case of an equicontinuity result that may be found in Appendix \ref{limmesdistrsec}. The reflection positivity of these limit measures (assuming the manifold is equipped with appropriate reflection positivity structure) is then proved using results from Section \ref{refposlimea}.

Section \ref{theinflimsec} contains our main construction described above, i.e. the scaling limit passage from the sphere to Euclidean space. A sequence of measures on $\mathcal{D}(\mathbb{S}^{d})$ is transferred to $\mathcal{D}(\mathbb{R}^{d})$ employing the natural unitary equivalences induced by  scaling transformations in $\mathbb{R}^{d}$ and the stereographic projection. The existence of a limit point is shown using the same equicontinuity result in Appendix \ref{limmesdistrsec} mentioned above and the Glimm-Jaffe axioms are verified.

In Appendix \ref{meahb} we briefly review for reader's convenience the basic theorems regarding probability measures on locally convex spaces that are used throughout the paper. In Appendix \ref{limmesdistrsec} we have collected several results related to equicontinuous sequences of characteristic functionals of measures and uniform integrability that are undoubtedly well-known but perhaps not as standard as the material in Appendix \ref{meahb}.

\vspace{7pt}

{\bf Acknowledgment.} The author is grateful to Wojciech Dybalski, Leonard Gross, Nikolay M. Nikolov and Yoh Tanimoto for many  helpful suggestions and stimulating discussions at various stages of this project.

\section{Reflection positivity}\label{refposlimea}
\subsection{Reflection positivity on Riemannian manifolds}\label{refposrimman}
Let $M$ be a $d$-dimensional connected Riemannian manifold. We denote by $\mathcal{D}(M)$ the space of the real-valued smooth compactly supported functions on $M$ equipped with the standard $ C^{\infty}$-topology and by $\mathcal{D}'(M)$ its topological dual equipped with the weak* topology, the space of distributions on $M$.

All probability measures on topological spaces considered in this article are {\em Radon measures}. The reader is referred to Appendix A for the necessary background on measures on locally convex spaces.

Given a measure $\mu$ on $\mathcal{D}'(M)$, we denote by $L^{p}(\mu)$ the Banach space of the {\em complex-valued} $L^{p}$ functions with respect to $\mu$. We write $L^{p}(M)$ for the  space of {\em real-valued} functions on $M$ which are $p$-summable with respect to the measure induced by the metric.

We assume that we are given a decomposition $M=M_{-}\cup M_{0} \cup M_{+}$ where $M_{0},M_{-}$ and  $ M_{+}$ are disjoint, $M_{0}$ is a closed submanifold, $M_{\pm}$ are open submanifolds, and an isometry $\Theta$ mapping diffeomorphically $M_{\pm}$ onto $M_{\mp}$. We  denote the induced unitary operator on $L^{2}(M)$ by $\Theta$ as well and for every open $U \subset M$ consider the following space of functionals on $\mathcal{D}'(M)$:
\begin{equation}\label{expalgu}\mathcal{A}_{U}=\Bigl\{ \Psi: \Psi(\phi)= \sum_{i}z_{i}e^{i\phi(f_{i})},\,  f_{i}\in \mathcal{D}(U),\, \, z_{i} \in \mathbb{C}\Bigr\}.
\end{equation}
Clearly $\mathcal{A}_{U}$ is an algebra under the operation of pointwise multiplication.

\begin{definition} (a) We say that a bounded positive operator $A$ on $L^{2}(M)$ is {\em reflection positive} with respect to
  $(M_{\pm}, \Theta)$ if $A \Theta = \Theta A$ and $\langle Af, \Theta f \rangle \geq 0$  for every $f \in L^{2}(M)$ that is supported in $M_{+}$.

(b) A functional $S$ on $\mathcal{D}(M)$ is {\em reflection positive} with respect to
  $(M_{\pm}, \Theta)$ if the matrix $S_{ij}=S(f_{i}-\Theta f_{j})$ is positive semi-definite for every finite sequence $f_{1},\ldots ,f_{k} \in \mathcal{D}(M)$ supported in $M_{+}$.

(c) A probability measure $\mu$ on $\mathcal{D}'(M)$ is {\em reflection positive}  with respect to
$(M_{\pm}, \Theta)$ if
\begin{equation}\label{refpos}
{\langle \Psi, \Theta \Psi \rangle}_{L^2(\mu)}\geq 0
\end{equation}
for all $\Psi \in \mathcal{A}_{M_{+}}$ (We denote the operator on $L^2(\mu)$ induced by $\Theta$ again by $\Theta$.)
\end{definition}

\begin{lemma}\label{refptwop}
(a) If an operator $A$ on $L^{2}(M)$ is reflection positive, the functional $S(f)=e^{-\langle Af, f \rangle}$ on $\mathcal{D}(M)$ is reflection positive as well.

(b)  A probability measure $\mu$ on $\mathcal{D}'(M)$ is reflection positive if and only if
its characteristic functional $S_{\mu}$ is reflection positive.

\end{lemma}
\begin{proof} (a)  See \cite[Theorem 6.2.2]{GJ}.

(b) Suppose first that  $\mu$ is reflection positive. We fix $f_{1},\ldots ,f_{k} \in \mathcal{D}(M_{+})$ and observe that
$$S_{ij}=S_{\mu}(f_{i}-\Theta f_{j})= \int_{\mathcal{D}'(M)}e^{ \phi(f_{i}-\Theta f_{j})} d\mu(\phi).$$
 Hence for every finite sequence
$z_{1},\ldots ,z_{k} \in \mathbb{C}$ one has
\begin{equation}\label{rein}
\sum_{i,j}z_{i}S_{ij}\bar{z}_{j}= \int_{\mathcal{D}'(M)}(\sum_{i}z_{i}e^{i\phi(f_{i})} )
\overline{(\Theta \sum_{i}z_{i} e^{i\phi(f_{i})})} d \mu(\phi)\geq 0
\end{equation}
which implies that $S_{\mu}$ reflection positive.

Conversely, suppose that $S_{\mu}$ is reflection positive. Then
Eq. (\ref{rein}) implies that  $\mu$ is reflection positive.
\end{proof}

\subsection{Reflection positive limits of measures}\label{reposlimimesec}

The main result in this subsection is a sufficient condition for the reflection positivity of limits of sequences of measures. For every $\tau > 0$ we denote by $M_{\pm}^{\tau}$ the set of all points in $M_{\pm}$ whose distance from $M_{0}$ is greater than $\tau$ and observe that the isometry $\Theta$ maps diffeomorphically $M^{\tau}_{\pm}$ onto $M^{\tau}_{\mp}$ so that one can define reflection positivity with respect to the data $(M^{\tau}_{\pm}, \Theta)$.

Let $\tau_{k} > 0$ be a decreasing sequence converging to 0 as $k \rightarrow \infty$ and let $\mu_{k}$ be a sequence of probability measures on $\mathcal{D}(M)$ reflection positive with respect to $(M^{\tau_{k}}_{\pm}, \Theta)$. Then the following continuity property of reflection positivity holds.

\begin{lemma}\label{contrepo}  Assume that the characteristic functionals $S_{\mu_{k}}$ of $\mu_{k}$ converge pointwise to the characteristic functional of a measure $\mu_{0}$. Then $\mu_{0}$ is reflection positive with respect to $(M_{\pm}, \Theta)$.
\end{lemma}
\begin{proof} By Lemma \ref{refptwop} (b) the characteristic functionals $S_{\mu_{k}}$ are reflection positive with respect to $(M_{\pm}^{\tau_{k}},\Theta)$, i.e. the matrices with entries
$$S^{\,k}_{ij}=S_{\mu_{k}}(f_{i}-\Theta f_{j})$$
are positive semi-definite for every finite sequence $f_{1},\ldots, f_{n} \in \mathcal{D}(M_{+}^{\tau_{k}})$. Now given a sequence $f_{1},\ldots, f_{n} $ supported in $M_{+}$, one can find sufficiently large $N$ so that $f_{1},\ldots, f_{n}$ are supported in $M_{+}^{\tau_{k}}$ for all $k >N$, which implies that $\mu_{0}$ is reflection positive.
\end{proof}

\begin{remark}\label{remonre}
Let $U$ be an open subset of $M$ and denote by $r_{U}: \mathcal{D}(M) \rightarrow C^{\infty}(U)$  the natural restriction map.  Given a probability measure $\mu$ on $\mathcal{D}(M)$, we define $\mathcal{R}_{\mu}(U)$ to be the set consisting of those functionals $\Psi \in  L^{2}(\mu)$ for which there exists  a functional $\overline{\Psi}\in L^{2}(C^{\infty}(U), r_{U}^{\bullet}\mu)$ such that
$\Psi=\overline{\Psi}r_{U}$ (here $r_{U}^{\bullet}\mu$ the denotes the image of $\mu$ under $r_{U}$.)

Then  a probability measure $\mu$ on $\mathcal{D}(M)$ is reflection positive with respect to $(M^{\tau}_{\pm}, \Theta)$ if and only if (\ref{refpos}) holds for every $\Psi \in \mathcal{R}_{\mu}(M^{\tau}_{+})$. This follows from the fact that $\mathcal{A}_{M^{\tau}_{+}}$ is dense in $L^{2}(\mathcal{D}'(M^{\tau}_{+}), r^{\bullet}_{M^{\tau}_{+}}\mu)$ (see e.g. \cite[Corollary 7.12.2]{B}).
\end{remark}

Let  $\rho_{k}$ be a sequence of densities on $\mathcal{D}(M)$, we define probability measures
$$ \mu^{\rho}_{k}=N_{k}\rho_{k}\mu_{k}$$
where $N_{k}=\|\rho_{k}\|^{-1}_{L^{1}(\mu_{k})} $ are normalization constants. In what follows we suppose that the sequence $N_{k}$ is bounded.

Let $N$ be a $d$-dimensional connected Riemannian manifold with reflection positivity data $(N_{\pm}, \Theta)$ defined as in Section \ref{refposrimman}. (For notational simplicity we denote reflections on different manifolds by the same letter $\Theta$.) Let further $I_{k}: \mathcal{D}(N) \rightarrow \mathcal{D}(M)$ be a sequence of operators commuting with the reflections $\Theta$ and with the property that there exists a sequence $\tau'_{k} > 0$ decreasing to 0 such that $I_{k}$ maps functions supported in $N^{\tau'_{k}}_{\pm}$ to functions supported in $M^{\tau_{k}}_{\pm}$. Finally, suppose that the characteristic functionals  $S_{\mu^{\rho}_{k}}\circ I_{k}$ converge pointwise to the characteristic functional of a measure $\mu^{\rho}_{0}$ on $\mathcal{D}'(N)$.

\begin{lemma}\label{repowdenlem} Assume that $\rho_{k}$ can be factored as a product of densities \begin{equation}\label{factrho}\rho_{k}=\rho^{+}_{k}\,\rho^{0}_{k}\,\Theta(\rho^{+}_{k})
\end{equation}
where $\Theta(\rho^{+}_{k})$ stands for the natural action of  $\hspace{2pt}\Theta$ on $\rho^{+}_{k}$, $\rho^{+}_{k}$ belongs to $\mathcal{R}_{\mu_{k}}(M^{\tau_k}_{+})$ defined in Remark \ref{remonre} and one has
\begin{equation}\label{contonrhop}
\lim_{k \rightarrow \infty}\|\rho^{0}_{k}-1\|_{L^2(\mu_{k})}=0,
\end{equation}
\begin{equation}\label{contonrhozero}
\|\rho^{+}_{k}\Theta(\rho^{+}_{k}))\|_{L^2(\mu_{k})}< K
\end{equation}
for some constant $K>0$.

Then $\mu^{\rho}_{0}$ is reflection positive with respect to $(N_{\pm}, \Theta)$.
\end{lemma}
\begin{proof} We consider the sequence of probability measures
$$ \bar{\mu}^{\rho}_{k}=\bar{N}_{k}\rho^{+}_{k}\Theta(\rho^{+}_{k})\mu_{k},$$
where $\bar{N}_{k}$ are normalization constants. We shall first show that the measures $\bar{\mu}^{\rho}_{k}$ are reflection positive with respect to $(M^{\tau_{k}}_{\pm}, \Theta)$ and then that the sequences $S_{\mu^{\rho}_{k}}\circ I_{k}$ and $S_{\bar{\mu}^{\rho}_{k}}\circ I_{k}$ have the same limit.

We fix $k$ and note that for every $\Phi \in \mathcal{R}_{\bar{\mu}^{\rho}_{k}}(M^{\tau_k}_{+})$ one has
$$\langle \Phi,\Theta \Phi \rangle_{L^2(\bar{\mu}^{\rho}_{k})} = \bar{N}_{k} \langle  \Phi\rho^{+}_{k},\Theta (\Phi \rho^{+}_{k})\rangle_{L^{2}(\mu_{k})}.$$
Since $\mu_{k}$ is reflection positive with respect to $(M^{\tau_{k}}_{\pm}, \Theta)$ and $\rho^{+}_{k} \in \mathcal{R}_{\mu_{k}}(M^{\tau_k}_{+})$, by Remark \ref{remonre} the latter quantity is non-negative, hence $\bar{\mu}^{\rho}_{k}$ is reflection positive with respect to $(M^{\tau_{k}}_{\pm}, \Theta)$ and  $S_{\bar{\mu}^{\rho}_{k}}\circ I_{k}$ is reflection positive with respect to $(N^{\tau'_{k}}_{\pm}, \Theta)$ by our assumptions on $I_{k}$.

Next we observe that  the Cauchy-Schwartz inequality implies
$$ |N_{k}^{-1}  - \bar{N}_{k}^{-1}|\leq \int_{\mathcal{D}(M)}(\rho^{+}_{k}\Theta(\rho^{+}_{k}))^{2}d\mu_{k}\int_{\mathcal{D}(M)}|\rho^{0}_{k}-1|^{2}d\mu_{k}, $$
hence by (\ref{contonrhop}), (\ref{contonrhozero}) and the boundedness of $N_{k}$ one has
\begin{equation}\label{nknkbarze} N_{k}  - \bar{N}_{k} \rightarrow 0
\end{equation}
as $k \rightarrow \infty$. Similarly one finds that

\begin{equation}\label{intnkbarze}
(\bar{N}^{-1}_{k}S_{\bar{\mu}^{\rho}_{k}}\circ I_{k}  - N^{-1}_{k}S_{\mu^{\rho}_{k}}\circ I_{k})(g)  \rightarrow 0
\end{equation}
for every $g \in \mathcal{D}(N)$. Now (\ref{nknkbarze}) and  (\ref{intnkbarze}) together with the boundedness of  $N_{k}$ imply that
$$ S_{\bar{\mu}^{\rho}_{k}}\circ I_{k}  - S_{\mu^{\rho}_{k}}\circ I_{k}  \rightarrow 0  ,$$
pointwise, hence the reflection positivity of $\mu^{\rho}_{0}$  with respect to $(N_{\pm}, \Theta)$ follows from Lemma \ref{contrepo} applied to the sequence of measures with characteristic functionals $ S_{\bar{\mu}^{\rho}_{k}}\circ I_{k}$.
\end{proof}

We now assume that $M$ is compact and that the densities $\rho_{k}$ have the form
\begin{equation}\label{genforofdens}
\rho_{k}=\exp\Bigl(\int_{M}F_{k}(f(x))dx\Bigr), \quad f \in \mathcal{D}(M),
\end{equation}
where $F_{k}: \mathbb{R} \rightarrow  \mathbb{R}$ are continuous functions and the integral over $M$ is with respect to the volume element induced by the metric.

We further suppose that the sequence of characteristic functionals $S_{\mu_{k}}$ of $\mu_{k}$, considered as measures on $\mathcal{D}'(M)$ (cf. Remark \ref{remonincmes}), is equicontinuous on $\mathcal{D}(M)$.  Then Lemma \ref{repowdenlem} and the results in Appendix \ref{uniinteggrasubs} allow  us to formulate the following criterion for reflection positivity.

\begin{proposition}\label{refposmacrit} Suppose that
\begin{equation}\label{fkcovtoze}
\lim_{k \rightarrow \infty} \int_{M\setminus (M^{\tau_{k}}_{+}\cup  M^{\tau_{k}}_{-})}F_{k}(f(x))dx = 0, \quad f \in \mathcal{D}(M),
\end{equation}
and that there exists $p>2$ such that
\begin{equation}\label{suppgretwo}\sup_{k}\|\rho_{k}\|_{L^{p}(\mu_{k})} < \infty.
\end{equation}
Then $\mu^{\rho}_{0}$ is reflection positive with respect to $(N_{\pm}, \Theta)$.
\end{proposition}

\begin{proof} We write
$$\rho_{k}=\rho^{+}_{k}\,\rho^{0}_{k}\Theta(\rho^{+}_{k}),   $$
where
$$ \rho^{+}_{k}(f)=\exp\Bigl(\int_{M^{\tau_{k}}_{+}}F_{k}(f(x))dx\Bigr), \quad f \in \mathcal{D}(M),$$
and
$$ \rho^{0}_{k}(f)= \exp\Bigl(\int_{M\setminus (M^{\tau_{k}}_{+}\cup  M^{\tau_{k}}_{-})}F_{k}(f(x))dx\Bigr), \quad f \in \mathcal{D}(M).$$
We shall show that this factorization of $\rho_{k}$ satisfies the conditions stated in Lemma \ref{repowdenlem}. We first observe that (\ref{suppgretwo}) implies $\sup_{k}\|\rho_{k}\|_{L^{2}(\mu_{C_{k}})} < \infty$, hence (\ref{contonrhozero}) holds for $\rho^{+}_{k}$.

Let us denote by $\overline{\rho}^{0}_{k}$ the extension of $\rho^{0}_{k}$ to $\mathcal{D}'(M)$ defined by setting $\overline{\rho}^{0}_{k}$ equal to 0 on $\mathcal{D}'(M)\setminus \mathcal{D}(M)$. A standard argument (see e.g. \cite[\S 1]{RR}) shows that the functionals $\overline{\rho}^{0}_{k}$ are Borel measurable. We note that by (\ref{fkcovtoze}) one has $\overline{\rho}^{0}_{k} \rightarrow 1_{\mathcal{D}(M)}$ pointwise, where $1_{\mathcal{D}(M)}$ is the indicator function of $\mathcal{D}(M)$. Thus, using the assumed equicontinuity of $S_{\mu_{k}}$, we conclude by Lemma \ref{gendomconvspcase} that $\overline{\rho}^{0}_{k} \rightarrow  1_{\mathcal{D}(M)}$ in $\mu_{k}$-measure. Further, by (\ref{suppgretwo}) and Lemma \ref{pgroneunile} the functionals $\rho^{2}_{k}$, and hence also $(\overline{\rho}^{0}_{k})^{2}$, are uniformly $\mu_{k}$-integrable. It now follows from Theorem \ref{genvitalithe} that (\ref{contonrhop}) holds for $\rho^{0}_{k}$, thus by Lemma \ref{repowdenlem} $\mu^{\rho}_{0}$ is reflection positive with respect to $(N_{\pm}, \Theta)$.
\end{proof}


\section{Finite volume measures}\label{finitevolmes}
In this section we assume that the Riemannian manifold $M$ considered in Section \ref{refposlimea} is closed (i.e. compact  without boundary) and, utilizing a sequence of mollifiers on $M$, present examples of the limit measure construction from Appendix \ref{equilimisubs} that turn out to be isometry invariant and reflection positive.
\subsection{Isometry invariant mollifiers}\label{isomolli}
We denote the compact group of isometries of $M$ by $G$ and observe that $G$ acts naturally on $\mathcal{D}(M)$ and hence on $\mathcal{D}'(M)$. Further, there is a natural unitary representation of $G$ on $L^2(M)$. Given a neighborhood $U$ of the diagonal in $M \times M$, we denote by $d (U)$ the largest distance between a point in $U$ and the diagonal.


\begin{lemma}\label{isomoll} There exists a sequence of smoothing continuous operators $A_{k}: \mathcal{D}'(M) \rightarrow \mathcal{D}(M)$ satisfying:

(1) The restriction of $A_{k}$ to $L^2(M)$ is a trace class operator and one has
$$\text{Tr}A_{k}= \int_{M} A_{k} (x,x) dx,$$
where $A_{k} (x,y)$ stands for the smooth integral kernel of $A_{k}$;

(2) The operators $A_{k}$ are uniformly bounded on $L^2(M)$;

(3) $A_{k} \rightarrow 1$ strongly on $L^2(M)$ as $k \rightarrow \infty$;

(4) $d (\text{supp}(A_{k} (x,y)))$ decreases to 0 as $k \rightarrow \infty$;

(5) $A_{k}$ commutes with the action of $G$ on $L^2(M)$.
\end{lemma}
\begin{proof} The existence of mollifiers $B_{k}$ satisfying (1)-(4) is a standard result (see e.g. \cite[Chapter II.7]{T}). To obtain isometry invariant mollifiers, we average over $G$ as follows. We set
$$ A_{k}(x,y)=\int_{G} B_{k}(g^{-1}x, g^{-1}y) dg,$$
where $dg$ denotes the normalized Haar measure on the compact group $G$. Then an application of Fubini's theorem shows that the operators $A_{k}$ with kernels  $A_{k}(x,y)$ satisfy (1) and (2). Similarly, (3) holds by Lebesgue's dominated convergence theorem. Finally, (4) is true since $G$ is a group of isometries.
\end{proof}

\subsection{Examples of limit measures}\label{isoivmesec} We set $C=(m^{2}+\triangle)^{-1}$, where $\triangle$ stands for the unique selfadjoint extension of the Laplacian on $L^{2}(M)$ and $m>0$ (we use the sign convention in which $\triangle$ is a nonnegative operator.) By Example \ref{gauhilex} the operator $C$ is the covariance of a unique mean zero Gaussian probability measure $\mu_{C}$ on $\mathcal{D}'(M)$ whose characteristic functional is given by (\ref{gauchar}).

Given a sequence of mollifiers $A_{k}$ satisfying the properties listed in Lemma \ref{isomoll}, we set $C_{k}=A^{*}_{k}CA_{k}$. Since the product of any pseudodifferential operator and a smoothing operator is a smoothing operator (see e.g. \cite[Chapter II.4]{T}), the operators $C_{k}$ are smoothing, i.e. $C_{k}: \mathcal{D}'(M) \rightarrow \mathcal{D}(M)$. We denote the corresponding  mean zero Gaussian probability measures on $\mathcal{D}(M)$ given by Example \ref{gausmoex} by $\mu_{C_{k}}$.

\begin{example}\label{bintex} (Bounded interaction) We fix a bounded continuous function $F: \mathbb{R} \rightarrow \mathbb{R}$ and define a continuous density
\begin{equation}\label{fvme}
\rho_{F}(f)=\exp\Bigl(\int_{M}F(f(x))dx\Bigr), \quad f \in \mathcal{D}(M).
\end{equation}
We denote the corresponding probability measures given by $\rho_{F}$ and $\mu_{C_{k}}$ as in (\ref{mukdeffinvo}) by $\mu_{F,k}$. Using the boundedness of $F$ and the compactness of $M$, Proposition \ref{tightcrit}(b) (applied with $\mathcal{V}=\mathcal{D}(M)$, $M=N$ and $I_{k}=\text{Id}$) and Theorem \ref{levystheo} imply that the sequence $\mu_{F,k}$ contains a subsequence weakly convergent in $\mathcal{D}'(M)$. We fix such a subsequence and denote its limit by  $\mu_{F}$.
\end{example}

\begin{remark} Let $\rho_{k}$ be densities on $\mathcal{D}(M)$ satisfying (\ref{uestfornn}). Assume further that  $\rho_{k}$  converges pointwise to 1 as $k \rightarrow \infty$. Then by Lemmas \ref{gendomconvspcase} and \ref{pgroneunile} $\rho_{k}$ are uniformly $\mu_{C_{k}}$-integrable and $\rho_{k} \rightarrow 1 $ in $\mu_{C_{k}}$-measure, thus we conclude by Theorem \ref{genvitalithe} that $\|\rho_{k}\|_{L^{1}(\mu_{C_{k}})}\rightarrow 1 $. Hence one expects that the limit measures obtained from $\rho_{k}$ via Proposition \ref{tightcrit} would coincide with $\mu_{C}$ in this case.
\end{remark}

\begin{example}\label{wickpow} ($P(\phi)_{2}$ models on the 2-sphere)
We define the $n$-th Wick power of $f \in \mathcal{D}(M)$ with respect to $C_{k}$ by setting (cf. \cite[ Eq. (8.5.5)]{GJ})
$$:\hspace{-3pt}f(x)^{n}\hspace{-3pt}:_{C_{k}}=\sum_{j=0}^{[n/2]}\frac{(-1)^{j}n!}{(n-2j)!j!2^{j}}C_{k}(x,x)^{j}f(x)^{n-2j},$$
where $C_{k}(x,y)$ is the integral kernel of the smoothing operator $C_{k}$ and $[\hspace{3pt} \cdot \hspace{3pt} ]$ denotes the integer part of a number.

More generally, for every bounded from below polynomial $P$, one can consider the corresponding polynomial in the Wick powers which we denote by $:\hspace{-2pt}P(f)\hspace{-2pt}:_{C_{k}}$. Further, we define  a sequence of functionals
$$\rho_{P,k}(\phi)= e^{-:P(\phi):_{C_{k}}}, \quad \phi \in \mathcal{D}'(M),$$
where
$$ :\hspace{-2pt}P(\phi)\hspace{-2pt}:_{C_{k}}= \int_{M}:\hspace{-2pt}P(A_{k}\phi)\hspace{-2pt}:_{C_{k}}$$
and probability measures
$\mu_{P,k}=N_{k} \rho_{P,k} \,\mu_{C} $ on  $\mathcal{D}'(M)$,
where $N_{k}$ are normalization constants.

When $M$ is the standard two-dimensional sphere $\mathbb{S}^2$ one can prove, proceeding exactly as in \cite[Chapter 8]{GJ}, that $:\hspace{-2pt}P(\phi)\hspace{-2pt}:_{C_{k}}$ converge in $L^2(\mu_{C})$ to a limit $:\hspace{-4pt}P(\phi)\hspace{-4pt}:_{C} $ as $k\rightarrow \infty$ and that $e^{-:P(\phi):_{C}}$ belongs to $L^1(\mu_{C})$ (see also \cite[Section 11.1]{BJM}). Using this, one can show as in the proof of Eq. (0.8) in \cite{C} that
$$  \rho_{P,k} \rightarrow e^{-:P(\phi):_{C}} $$
in $L^p(\mu_{C})$ for $1\leq p <  \infty$. It follows that $\rho_{P,k}$ satisfy the assumptions of Proposition \ref{tightcrit} with respect to $C$, hence there exists a subsequence of $\mu_{P,k}$ converging weakly to a measure  $\mu_{P}$.
\end{example}
\begin{remark} We conjecture that an alternative construction of the $P(\phi)_{2}$ models may be obtained employing the sequence $\mu_{C_{k}}$. More precisely, we conjecture that the densities
$$ \widetilde{\rho}_{P,k}= \exp \Bigl(-\int_{M}:\hspace{-2pt}P(f)\hspace{-2pt}:_{C_{k}}dx\Bigr), \quad f \in \mathcal{D}(M) $$
 satisfy the assumptions of Proposition \ref{tightcrit} with respect to $C_{k}$ and the resulting limit measures of the sequence $ \widetilde{\mu}_{P,k}=N_{k} \widetilde{\rho}_{P,k}\,\mu_{C_{k}} $ are equivalent in appropriate sense to $\mu_{P}$.
\end{remark}

Following \cite{JR}, we now impose an additional restriction on the reflection positivity data ($M_{\pm}, \Theta$) considered in Section \ref{refposrimman}. We suppose that $\Theta$ fixes the points in  $M_{0}$ and induces hyperplane reflections on the tangent spaces at points in $M_{0}$. This assumption ensures that the covariance $C$ is reflection positive  with respect to ($M_{\pm}, \Theta$) (cf. \cite[Theorem 1]{JR}).

\begin{proposition} The measures  $\mu_{F}$ and $\mu_{P}$ on $\mathcal{D}'(M)$ constructed in Examples \ref{bintex} and \ref{wickpow} respectively are isometry invariant and reflection positive with respect to ($M_{\pm}, \Theta$).
\end{proposition}
\begin{proof} Since the operators $C$ and $A_{k}$ are isometry invariant, so are the measures $\mu_{C}$ and $\mu_{C_{k}}$. This, together with the isometry invariance of the densities constructed in  Examples \ref{bintex} and \ref{wickpow} implies the isometry invariance of the measures $\mu_{F}$ and $\mu_{P}$ on $\mathcal{D}'(M)$.

Since one has
$$ \langle C_{k}f, \Theta f \rangle =\langle CA_{k}f, \Theta  A_{k}f\rangle , \quad f \in L^{2}(M),$$
it follows from Lemma \ref{isomoll}(4) and the reflection positivity of $C$ that there exists a sequence $\tau_{k} >0$  decreasing to zero such that $C_{k}$ and hence $\mu_{C_{k}}$ is reflection positive with respect to  $(M^{\tau_{k}}_{\pm}, \Theta)$. Thus by Proposition \ref{refposmacrit} (taking $\mu_{k}=\mu_{C_{k}}$, $M=N$ and $I_{k}=\text{Id}$) the  measure $\mu_{F}$ is reflection positive  with respect to ($M_{\pm}, \Theta$). We note that Proposition \ref{refposmacrit} does not imply the reflection positivity of $\mu_{P}$, however the latter is well-known (see e.g. \cite[Theorem 10.4.3]{GJ}).
\end{proof}

\section{Infinite volume measures}\label{theinflimsec}

\subsection{Conformal maps, Laplacians, and the stereographic projection}
Let $(M, g_{M})$ and $(N, g_{N})$ be a pair of $d$-dimensional Riemannian manifolds and let $\eta: M \rightarrow N$ be a conformal diffeomorphism with conformal factor $\Lambda_{\eta}\in C^{\infty}(M)$, i.e.  $\eta^{*}g_{N}=\Lambda^{2}_{\eta}\hspace{2pt}g_{M}$ (Here and below upper $*$ stands for the pull-back map acting on functions and tensors.) Then the operator $U_{\eta}$ given by
\begin{equation}\label{uniconfo}
U_{\eta}f=\Lambda^{d/2}_{\eta}\eta^{*}f, \quad f \in \mathcal{D}(N)
\end{equation}
extends to a unitary operator from $L^{2}(N)$ to $L^{2}(M)$ such that $U^{-1}_{\eta}=U_{\eta^{-1}}$.

We denote by $\mathbb{S}^{d}$ the $d$-dimensional unit sphere in $\mathbb{R}^{d+1}$ centered at the origin and set  $\mathbb{S}_{p}^{d}= \mathbb{S}^{d}\setminus\{(0,0,\ldots, 0, 1)\}$. We regard $\mathbb{S}^{d}$ and $\mathbb{S}_{p}^{d}$ as  Riemannian manifolds with respect to the metric induced by the standard metric on $\mathbb{R}^{d+1}$ and in what follows identify $L^{2}(\mathbb{S}_{p}^{d})$ with $L^{2}(\mathbb{S}^{d})$.

Let $\alpha: \mathbb{S}_{p}^{d} \rightarrow \mathbb{R}^{d}$ be the stereographic projection map with  its $i$-th component given by

\begin{equation}\label{sterpro}
\alpha_{i}(x)=\frac{x_{i}}{1-x_{d}}, \quad i=0,\ldots, d-1,
\end{equation}
where $x=(x_{0},\ldots, x_{d}) \in \mathbb{R}^{d+1}$. We note that $\alpha$ is conformal with $\Lambda_{\alpha}=(1-x_{d})^{-1}$ and $\Lambda_{\alpha^{-1}}=\sqrt{2}(1+\|y\|^2)^{-1/2}$, where $y \in \mathbb{R}^{d}$.

Another example of a conformal map is provided by the scaling map $$\beta_{k}y=ky$$ on $\mathbb{R}^{d}$ defined for every $k>0$; clearly one has $\Lambda_{\beta_{k}}=k$.

A second order differential operator $D$ on a Riemannian manifold on $M$ is said to be {\em conformally covariant} if for every conformal $\eta: M \rightarrow M$ one has
$$D=U^{-1}_{\eta}\Lambda^{-1}_{\eta}D \Lambda^{-1}_{\eta}U_{\eta},$$
where the powers of $\Lambda_{\eta}$ act by multiplication. It is well-known that the Laplacian $\triangle_{E}$ on $\mathbb{R}^{d}$ is conformally covariant. Another example of a conformally covariant operator is the {\em conformal Laplacian} (also known as the Yamabe operator) on the sphere defined by $\triangle_{S}^{c}=\triangle_{S}+d(d-2)/4$, where $\triangle_{S}$ is the Laplacian on $\mathbb{S}^{d}$. The conformal covariance of $\triangle^{S}_{c}$ implies (see e.g. \cite{Gr})
\begin{equation}\label{confconn}
\triangle_{E}= U^{-1}_{\alpha}\Lambda^{-1}_{\alpha} \triangle_{S}^{c}\Lambda^{-1}_{\alpha}U_{\alpha}.
\end{equation}

\subsection{The scaling limit} In this subsection we present our main construction which involves taking simultaneously the ultraviolet and infrared (=infinite volume) limit in order to obtain probability measures on $\mathcal{D}'(\mathbb{R}^{d})$. We refer to this limiting procedure as a {\em scaling limit} because of the use the scaling maps $\beta_{k}$. We note that somewhat similar scaling limit has been described in \cite[Section 3]{Sc2} in the context of Algebraic QFT.

We define a sequence of covariance operators on $L^{2}(\mathbb{S}^{d})$  via
\begin{equation}\label{comforeso}
C_{S,k}= U_{\alpha}U_{\beta_{k}}(\triangle_{E}+m^{2})^{-1}U_{\beta_{k}}^{-1}U_{\alpha}^{-1}, \quad k \in \mathbb{N}.
\end{equation}
The conformal covariance of $\triangle_{E}$ implies
\begin{equation}\label{confsqro}
U_{\beta_{k}}(\triangle_{E}+m^{2})^{-1}U_{\beta_{k}}^{-1}= k^{2}(\triangle_{E}+k^{2}m^{2})^{-1}.
\end{equation}
Using (\ref{confconn}) one finds that
$$ \triangle_{E}+k^{2}m^{2}=U_{\alpha}^{-1}\Lambda^{-1}_{\alpha}(\triangle_{S}^{c}+k^{2}m^{2}
\Lambda^{2}_{\alpha})\Lambda^{-1}_{\alpha}U_{\alpha}$$
which together with (\ref{confsqro}) easily implies
\begin{equation}\label{ceskintco}
C_{S,k}=k^{2}\Lambda_{\alpha}(\triangle_{S}^{c}+k^{2}m^{2}\Lambda^{2}_{\alpha})^{-1}\Lambda_{\alpha}.
\end{equation}
We set
$$\widetilde{C}_{S,k}= A_{k}^{*}C_{S,k} A_{k},$$
where $A_{k}$ is a sequence of mollifiers on $\mathcal{D}'(\mathbb{S}^{d})$ satisfying the properties listed in Lemma \ref{isomoll}. The operators $\widetilde{C}_{S,k}$ are smoothing, hence they define centered Gaussian probability measures $\mu_{\widetilde{C}_{S,k}}$ on $\mathcal{D}(\mathbb{S}^{d})$ (cf. Example \ref{gausmoex}).

Now assume that we are given a sequence of densities $\rho^{S}_{k}$ on $\mathcal{D}(\mathbb{S}^{d})$ satisfying the estimates
(\ref{uestfornn}) with respect to the sequence $\mu_{\widetilde{C}_{S,k}}$, i.e. one has
\begin{equation}\label{rhosoneest}
\inf_{k}\|\rho^{S}_{k}\|_{L^{1}(\mu_{\widetilde{C}_{S,k}})} >0, \quad \sup_{k}\|\rho^{S}_{k}\|_{L^{2}(\mu_{\widetilde{C}_{S,k}})} < \infty.
\end{equation}
In particular, these conditions are satisfied by the sequences of densities defined in (\ref{fvme}) setting $M=\mathbb{S}^{d}$.
We define a sequence of measures on $\mathcal{D}(\mathbb{S}^{d})$ via
$$ \mu^{S,\rho} _{k}=N^{S}_{k}\rho^{S}_{k}\mu_{\widetilde{C}_{S,k}},$$
where $N^{S}_{k}$ are normalization constants.

\begin{proposition}\label{dzelimfu} (a) The sequence of functionals $S_{\mu^{S,\rho} _{k}}\circ U_{\alpha}U_{\beta_{k}} $ on $\mathcal{D}( \mathbb{R}^{d})$ is equicontinuous.

(b) There exists a subsequence of $S_{\mu^{S,\rho} _{k}}\circ U_{\alpha}U_{\beta_{k}} $ converging to the characteristic functional $S_{\mu^{E,\rho}}$ of a probability measure $\mu^{E,\rho}$ on $\mathcal{D}'(\mathbb{R}^{d})$.
\end{proposition}
\begin{proof}
Since the operators $\widetilde{C}_{S,k}$ are uniformly $L^{2}$-bounded by (\ref{comforeso}) and Lemma \ref{isomoll}(2), the result follows from Proposition \ref{tightcrit} taking $M=\mathbb{S}^{d}$, $N=\mathbb{R}^{d}$, $C_{k}=\widetilde{C}_{S,k}$ and $I_{k}=U_{\alpha}U_{\beta_{k}}$.
\end{proof}

\subsection{Euclidean invariance}\label{euclinvarse}
We identify the group of isometries of $\mathbb{S}^{d}$ with $O(d+1)$ and the subgroup of isometries preserving the point $(0,0,\ldots, 0, 1)$ with $O(d)$, thus obtaining a natural action of  $O(d)$ on $\mathbb{S}^{d}$. Let $\{e_{0}, e_{1},\ldots, e_{d}\}$ be the standard orthonormal basis of $\mathbb{R}^{d+1}$. For $i=0, \ldots, d-1$ we denote by $L_{i}$ the generator of the one-parameter group of rotations in the plane spanned by $e_{i}$ and  $e_{d}$, considered as a element of the Lie algebra of $O(d+1)$. For every  translation $T=(t_{0},\ldots, t_{d-1})\in \mathbb{R}^{d}$  we set
$$g_{k}(T)=\exp \bigl(2k^{-1}\sum_{j}t_{j}L_{j}\bigr) \in O(d+1).$$
Further, we define
$$T_{E}(x)=x+T, \quad x \in \mathbb{R}^{d},$$
$$T_{S,k}(y)=g_{k}(T)\cdot y, \quad y \in \mathbb{S}^{d}.$$
For the following lemma, recall that $*$ denotes the induced map on functions.
\begin{lemma}\label{euinvlem} For every $f \in \mathcal{D}( \mathbb{R}^{d})$ one has
\begin{equation}\label{equsttstar}\lim_{k \rightarrow \infty}\|T^{*}_{S,k}U_{\alpha}U_{\beta_{k}}f- U_{\alpha}U_{\beta_{k}}T^{*}_{E}f \|_{L^2(\mathbb{S}^{d})}=0\end{equation}
and
\begin{equation}
\label{equstraapr}\lim_{k\rightarrow \infty}(S_{\mu^{S,\rho}_{k}}(T^{*}_{S,k}U_{\alpha}U_{\beta_{k}}f)-S_{\mu^{S\rho}_{k}}(U_{\alpha}U_{\beta_{k}}T^{*}_{E}f))=0.
\end{equation}
\end{lemma}

\begin{proof} Using the unitarity of $U_{\alpha}$ and $U_{\beta_{k}}$, we see that that (\ref{equsttstar}) is equivalent to
\begin{equation}\label{chavaroog} \|U_{\beta_{k}}^{-1}U_{\alpha}^{-1}T^{*}_{S,k}U_{\alpha}U_{\beta_{k}}f- T^{*}_{E}f \|_{L^2(\mathbb{R}^{d})} \rightarrow 0.
\end{equation}
It is easy to check that
$$ \|T^{*}_{S,k}\Lambda^{d/2}_{\alpha}\alpha^{*} U_{\beta_{k}}f-\Lambda^{d/2}_{\alpha}T^{*}_{S,k}\alpha^{*} U_{\beta_{k}}f \|_{L^2(\mathbb{S}^{d})} \rightarrow 0  ,$$
which implies that to prove (\ref{chavaroog}) it suffices to show that
\begin{equation}\label{chavartt} \|U_{\beta_{k}}^{-1}(\alpha^{*})^{-1}T^{*}_{S,k}\alpha^{*}U_{\beta_{k}}f- T^{*}_{E}f \|_{L^2(\mathbb{R}^{d})} \rightarrow 0.
\end{equation}
 We shall establish (\ref{chavartt})  by showing that
\begin{equation}\label{poinestain}\lim_{k \rightarrow \infty}k\hspace{2pt}\alpha(g_{k}(T)\cdot \alpha^{-1}(k^{-1}x))=x+T
\end{equation}
for every $x \in \mathbb{R}^{d}$. For notational simplicity we verify (\ref{poinestain}) below for $d=1$ only; the proof in the general case is similar.

Recall that the $i$-th component of $\alpha^{-1}$ is given by
\begin{equation}\label{sterproinv}
\alpha^{-1}_{i}(y)=\frac{2y_{i}}{\|y\|^{2}+1}, \quad i=0,\ldots, d-1,
\quad
\alpha^{-1}_{d}(y)=\frac{\|y\|^{2}-1}{\|y\|^{2}+1},
\end{equation}
where $y=(y_{0},\ldots, y_{d-1}) \in \mathbb{R}^{d}$.
A straightforward computation utilizing (\ref{sterpro}) and (\ref{sterproinv}) shows that for every $x \in \mathbb{R}$ one has
$$k\hspace{2pt}\alpha (g_{k}(T)\cdot \alpha^{-1}(k^{-1}x))
=\frac{\frac{2x}{k^{-2}x^{2}+1}\cos\frac{2T}{k}-\frac{k^{-2}x^{2}-1}{k^{-2}x^{2}+1}\cdot k\sin\frac{2T}{k}}{1-\frac{2k^{-1}x}{k^{-2}x^{2}+1}\sin\frac{2T}{k}-\frac{k^{-2}x^{2}-1}{k^{-2}x^{2}+1}\cos\frac{2T}{k}},
$$
which easily implies (\ref{poinestain}). Since the integration in (\ref{chavartt}) is over a compact subset of  $\mathbb{R}^{d}$, an application of Lebesgue's dominated convergence theorem shows that (\ref{chavartt}) follows from (\ref{poinestain}) thus establishing (\ref{equsttstar}).

We now observe that (\ref{rhosoneest}), (\ref{gaudiff}) and the uniform boundedness of $\widetilde{C}_{S,k}$  imply
$$|S_{\mu^{S,\rho}_{k}}(T^{*}_{S,k}U_{\alpha}U_{\beta_{k}}f)-S_{\mu^{S,\rho}_{k}}(U_{\alpha}U_{\beta_{k}}T^{*}_{E}f)|$$
$$\leq
K_{1} \bigl(1-\exp \bigl(-K_{2}\|T^{*}_{S,k}U_{\alpha}U_{\beta_{k}}f- U_{\alpha}U_{\beta_{k}}T^{*}_{E}f \|^{2}_{L^2(\mathbb{S}^{d})}\bigl)\bigr) ,$$
for some positive constants $K_{1}$ and $K_{2}$. Thus, using one  (\ref{equsttstar}), we conclude that
(\ref{equstraapr}) holds.
\end{proof}

\begin{lemma}\label{lemmonapp} One has
\begin{equation}\label{appgkequ}\lim_{k\rightarrow \infty}(S_{\mu^{S,\rho}_{k}}(T^{*}_{S,k}U_{\alpha}U_{\beta_{k}}f)-S_{\mu^{S,\rho}_{k}}(U_{\alpha}U_{\beta_{k}}f))=0
\end{equation}
for every $f \in \mathcal{D}( \mathbb{R}^{d})$.
\end{lemma}
\begin{proof}
Employing the Cauchy-Schwartz inequality and (\ref{rhosoneest}) one sees that
$$|S_{\mu^{S,\rho}_{k}}(T^{*}_{S,k}U_{\alpha}U_{\beta_{k}}f)-S_{\mu^{S,\rho}_{k}}(U_{\alpha}U_{\beta_{k}}f))|\leq
K \int_{\mathcal{D}( \mathbb{S}^{d})}\Phi_{k}(f,g) d\mu_{\widetilde{C}_{S,k}}(g)$$
for some positive constant $K$, where
$$\Phi_{k}(f,g)= \bigl|e^{i \langle g, T^{*}_{S,k}U_{\alpha}U_{\beta_{k}}f \rangle}-  e^{i\langle g, U_{\alpha}U_{\beta_{k}}f \rangle} \bigr|^{2}.$$
Since by definition $T^{*}_{S,k}$ converges strongly to 1 as $k \rightarrow \infty$, one has $\Phi_{k}(f,g)\rightarrow 0$ for all $f \in \mathcal{D}( \mathbb{R}^{d}), g \in \mathcal{D}( \mathbb{S}^{d})$. Let us fix $f$ and denote by $\Phi^{f}_{k}$  the extension of $\Phi_{k}(f,\cdot)$ to $\mathcal{D}'(\mathbb{S}^{d})$ defined by setting $\Phi^{f}_{k}$ equal to 0 on $\mathcal{D}'(\mathbb{S}^{d})\setminus \mathcal{D}(\mathbb{S}^{d})$. A standard argument (see e.g. \cite[\S 1]{RR}) shows that the maps $\Phi^{f}_{k}$ are Borel measurable.

Further, since $\widetilde{C}_{S,k}$ are uniformly bounded, by Proposition \ref{tightcrit} the sequence of  measures $\mu_{\widetilde{C}_{S,k}}$ on $\mathcal{D}'(\mathbb{S}^{d})$ has a sequence of characteristic functionals equicontinuous on $\mathcal{D}(\mathbb{S}^{d})$. Thus by Lemma  \ref{gendomconvspcase} $\Phi^{f}_{k}$ converges to 0 in $\mu_{\widetilde{C}_{S,k}}$-measure. Since $\Phi^{f}_{k}$ are uniformly bounded and hence  uniformly $\mu_{\widetilde{C}_{S,k}}$-integrable,  we conclude by Theorem \ref{genvitalithe} that (\ref{appgkequ}) holds.
\end{proof}

\begin{proposition}\label{eucinfu} If the densities $\rho^{S}_{k}$ are $O(d+1)$-invariant the limit measure $\mu^{E,\rho}$ constructed in Proposition \ref{dzelimfu} is invariant under the natural action of the Euclidean group $E(d)$.
\end{proposition}
\begin{proof}  Clearly the maps $\alpha$ and $U_{\alpha}$ are $O(d)$-equivariant (with respect to the action of $O(d)$ on $\mathbb{S}^{d}$ defined in the beginning of this subsection and the natural action of $O(d)$ on $\mathbb{R}^{d}$.) It follows, using the $O(d+1)$-invariance of $\rho^{S}_{k}$, that $S_{\mu^{S,\rho} _{k}}\circ U_{\alpha}U_{\beta_{k}}$ and hence $S_{\mu^{E,\rho}}$ are $O(d)$-invariant.

To show that $S_{\mu^{E,\rho}}$ is translation invariant, it suffices to prove that
given $T \in \mathbb{R}^{d}$, one has
$$ S_{\mu^{S,\rho}_{k}}(U_{\alpha}U_{\beta_{k}}T^{*}_{E}f) - S_{\mu^{S,\rho}_{k}}(U_{\alpha}U_{\beta_{k}}f) \rightarrow 0  .$$
But this follows directly from (\ref{appgkequ}) and (\ref{equstraapr}).
\end{proof}

\subsection{Verification of the Glimm-Jaffe axioms}\label{veriosaxi}
We begin by introducing the standard reflection positivity data on $\mathbb{R}^{d}$ and $\mathbb{S}^{d}$. We write $(t, x_{1}, \ldots, x_{d-1})$ for the coordinates of $x \in \mathbb{R}^{d}$ and  set
$$\mathbb{R}^{d}_{+}=\left\{t, x_{1}, \ldots, x_{d-1} \in \mathbb{R}^{d}|t > 0\right\}, $$
$$\mathbb{R}^{d}_{-}=\left\{t, x_{1}, \ldots, x_{d-1} \in \mathbb{R}^{d}|t < 0\right\}, $$
$$\mathbb{R}^{d}_{0}=\left\{t, x_{1}, \ldots, x_{d-1} \in \mathbb{R}^{d}|t = 0\right\}. $$
Further, we write $\Theta$ for the reflection
$$(t, x_{1}, \ldots, x_{d-1}) \mapsto (-t, x_{1}, \ldots, x_{d-1}).$$ We set $\mathbb{S}^{d}_{\pm}=\mathbb{S}^{d} \cap \mathbb{R}^{d+1}_{\pm}$ and write $\Theta$ also for the reflection on $\mathbb{S}^{d}\subset \mathbb{R}^{d+1}$ induced by the reflection with respect to the first coordinate in $\mathbb{R}^{d+1}$.

We now assume that the mollifiers $A_{k}$ obtained in Lemma \ref{isomoll} possess the following additional properties. First, we suppose that
\begin{equation}\label{fastrakco}
k\cdot d (\text{supp}(A_{k} (x,y))) \rightarrow 0.
\end{equation}
Clearly (\ref{fastrakco}) may be achieved simply by appropriately re-indexing a given sequence of mollifiers. Second, setting $ A^{E}_{k}=U_{\alpha}^{-1}A_{k}U_{\alpha} $, we assume that
\begin{equation}\label{molscacommu}
\|A^{E}_{k}U_{\beta_{k}} f-U_{\beta_{k}}f\|_{L^2} \rightarrow 0
\end{equation}
for every  $f \in \mathcal{D}(\mathbb{R}^{d})$. It follows from the standard local construction of mollifiers (see e.g. \cite[\S 2.29]{AF}), that it can be arranged that (\ref{molscacommu}) hold.

The following theorem asserts that under certain natural assumptions the scaling limit measure $\mu^{E,\rho}$ satisfies the Glimm-Jaffe axioms OS0--3 from \cite[Section 6.1]{GJ} except the local integrability condition in the regularity axiom OS1.

\begin{theorem}\label{theoglijaax}
Assume that $\mu^{E,\rho}$ is constructed from densities $\rho^{S}_{k}$ that satisfy (\ref{genforofdens})-(\ref{suppgretwo}). Then  $\mu^{E,\rho}$  has the following properties.

(1) $\mu^{E,\rho}$ is Euclidean invariant.

(2) $\mu^{E,\rho}$ is reflection positive with respect to $(\mathbb{R}^{d}_{\pm},\Theta)$.

(3) The extension $S_{\mu^{E,\rho}}^{\mathbb{C}}$ of $S_{\mu^{E}}$ (cf. Appendix \ref{compfanalysubsec}) is analytic and satisfies
\begin{equation}\label{regl2condifi}
|S^{\mathbb{C}}_{\mu^{E,\rho}}(f)|\leq K_{1}e^{K_{2}\|f\|^{2}_{L^2}}  ,\quad f \in \mathcal{D}(\mathbb{R}^{d})_{\mathbb{C}},
\end{equation}
for some positive constants $K_{1}$ and $K_{2}$.
\end{theorem}

\begin{proof}
Part (1) follows from Proposition \ref{eucinfu} since densities $\rho^{S}_{k}$ given by (\ref{genforofdens}) are $O(d+1)$-invariant. To prove Part (2) we note that since $U_{\alpha}U_{\beta_{k}}$ maps functions supported in $\mathbb{R}^{d}_{\pm}$ to functions supported in $\mathbb{S}^{d}_{\pm}$, the reflection positivity of $(\triangle_{E}+m^{2})^{-1}$ with respect to $(\mathbb{R}^{d}_{\pm},\Theta)$ implies the reflection positivity of
 $C_{S,k}$ with respect to $(\mathbb{S}^{d}_{\pm},\Theta)$. Using this and (\ref{fastrakco}), we see that  there exists a decreasing positive sequence $\tau_{k}$ such that $k \hspace{1pt} \tau_{k} \rightarrow 0$ and $\mu_{\widetilde{C}_{S,k}}$ is reflection positive with respect to $((\mathbb{S}^{d}_{\pm})^{\tau_{k}}, \Theta)$ (we use the notation $M^{\tau}_{\pm}$  introduced in the beginning of Section \ref{reposlimimesec}).

 Now observe that $U_{\beta_{k}}$ maps functions supported in  $(\mathbb{R}^{d}_{\pm})^{\tau}$ to functions supported in  $(\mathbb{R}^{d}_{\pm})^{\tau /k}$. It follows that there exists $\tau'_{k}>0$ decreasing to 0 such that $U_{\alpha}U_{\beta_{k}}$ maps functions supported in $(\mathbb{R}^{d}_{\pm})^{\tau'_{k}}$ to functions supported in $(\mathbb{S}^{d}_{\pm})^{\tau_{k}}$. Hence we can apply Proposition \ref{refposmacrit} with $M=\mathbb{S}^{d}$, $N=\mathbb{R}^{d}$ and $I_{k}=U_{\alpha}U_{\beta_{k}}$ to conclude that Part (2) holds.

To establish Part (3) we find, using (\ref{rhosoneest}) and the Cauchy-Schwartz inequality, that the extensions of  $S_{\mu^{S,\rho}_{k}}$ to $\mathcal{D}( \mathbb{R}^{d})_{\mathbb{C}}$  are analytic and satisfy
 $$ |S^{\mathbb{C}}_{\mu^{S,\rho}_{k}}(f)| \leq  K \left(\int_{\mathcal{D}'( \mathbb{S}^{d})}\exp(\phi(U_{\alpha}U_{\beta_{k}}\text{Im}f))d\mu_{\widetilde{C}_{S,k}}(\phi) \right)^{1/2}$$
 for some constant $K>0$. Using this, (\ref{extgauseti}) and the uniform boundedness of  $\widetilde{C}_{S,k}$, we conclude that
\begin{equation}\label{regl2condifikdep}
|S^{\mathbb{C}}_{\mu^{S,\rho}_{k}}(f)|\leq K_{1}e^{K_{2}\|f\|^{2}_{L^2}}  ,\quad f \in \mathcal{D}(\mathbb{R}^{d})_{\mathbb{C}},
\end{equation}
for some positive constants $K_{1}$ and $K_{2}$. It follows that the sequence $S^{\mathbb{C}}_{\mu^{S,\rho}_{k}}$ is uniformly locally bounded and hence equicontinuous by Theorem \ref{locbouequith}. Now Theorem \ref{genaath} implies that $S^{\mathbb{C}}_{\mu^{S,\rho}_{k}}$  contains a subsequence converging uniformly on compact subsets. Since the limit of this sequence is analytic, we conclude that $S_{\mu^{E,\rho}}$ has analytic extension that by virtue of (\ref{regl2condifikdep}) satisfies (\ref{regl2condifi}).
\end{proof}

\begin{example} The assumptions of Theorem \ref{theoglijaax} are satisfied when $\mu^{E,\rho}$ is constructed from the sequences of densities defined in (\ref{fvme}), setting $M=\mathbb{S}^{d}$.
\end{example}

\begin{theorem} Suppose that the measure $\mu^{E,\rho}$ satisfies the assumptions of Theorem \ref{theoglijaax}. Then $\mu^{E,\rho}$ defines a relativistic quantum field satisfying all Wightman axioms except possibly the uniqueness of the vacuum axiom.
\end{theorem}

\begin{proof} We observe that Theorem \ref{theoglijaax}(3) implies that the moments of all orders of $\mu^{E,\rho}$ exist and extend to continuous multilinear functionals on $L^{2}(\mathbb{R}^{d})\times \cdots \times L^{2}(\mathbb{R}^{d})$ (cf. \cite[Propositions 6.1.4 and 19.1.1]{GJ}). In what follows, we write
$$C_{\mu}(f,f)=\int \phi(f)^{2}d\mu(\phi)$$
for the (extended) second moment of a measure $\mu$.

It is proved in \cite[Chapter 19]{GJ} that the axioms OS0--3 imply the Wightman axioms except the uniqueness of the vacuum. An inspection of this proof reveals that the only place where the local integrability condition in the regularity axiom is used is Proposition 19.1.4 in \cite{GJ}.  However, the statement of this proposition holds if
\begin{equation}\label{otcondsecmo}
C_{\mu^{E,\rho}}(f,f)=o(t) \,\, \text{as} \,\, t \rightarrow 0,
\end{equation}
where $f=\chi_{(0,t)}\otimes h$, with $h \in \mathcal{D}(\mathbb{R}^{d-1})$, and $\chi_{(0,t)}$ is the indicator function of the interval $(0,t)$.

To verify (\ref{otcondsecmo}), we employ the Cauchy-Schwartz inequality and (\ref{comforeso}) to find that
$$C_{\mu^{S,\rho}_{k}}(U_{\alpha}U_{\beta_{k}}f,U_{\alpha}U_{\beta_{k}}f)$$
$$\leq \text{const}
\cdot \langle (U_{\alpha}U_{\beta_{k}}(\triangle_{E}+m^{2})^{-1}U_{\beta_{k}}^{-1}U_{\alpha}^{-1}A_{k}U_{\alpha}U_{\beta_{k}} f ,A_{k}U_{\alpha}U_{\beta_{k}}f \rangle $$
for every $k$ and every $f \in \mathcal{D}(\mathbb{R}^{d})$. Taking the limit $k   \rightarrow  \infty$ in the latter inequality we see, using (\ref{molscacommu}) and the polarization identity, that
$$ C_{\mu^{E,\rho}}(f,f) \leq
\text{const} \cdot \langle (\triangle_{E}+m^{2})^{-1} f ,f \rangle.$$
Since the covariance $(\triangle_{E}+m^{2})^{-1}$ is known to satisfy (\ref{otcondsecmo}), we conclude that (\ref{otcondsecmo}) holds for $C_{\mu^{E,\rho}}$ as well.
\end{proof}

\begin{remark} Using (\ref{comforeso}) and the unitarity of $U_{\alpha}$ and $U_{\beta_{k}}$, one finds that
\begin{equation}\label{freecanor}
\|C^{1/2}_{S,k}U_{\alpha}U_{\beta_{k}}f\|_{L^{2}(\mathbb{S}^{d})}=\| (\triangle_{E}+m^{2})^{-1/2} f \|_{L^{2}(\mathbb{R}^{d})}
\end{equation}
for every $k$ and every $f \in \mathcal{D}(\mathbb{R}^{d})$. Further, it follows from Lemma \ref{isomoll}(3) and (\ref{molscacommu}) that
\begin{equation}\label{diffctbar}
S_{\mu_{\widetilde{C}_{S,k}}}( U_{\alpha}U_{\beta_{k}}f)- S_{\mu_{C_{S,k}}}(U_{\alpha}U_{\beta_{k}}f) \rightarrow 0,
\end{equation}
 for every $f \in \mathcal{D}(\mathbb{R}^{d})$. Now  (\ref{freecanor}) and (\ref{diffctbar}) imply that the scaling limit $\mu^{E,\rho}$ in the case of a free scalar field coincides, as desired, with the Gaussian measure with covariance $(\triangle_{E}+m^{2})^{-1}$. This observation motivates the definition of the covariances $C_{S,k}$. One expects that if the scaling limit is taken using $(\triangle_{S}+m^{2})^{-1}$ instead of $C_{S,k}$ then this limit would be a massless theory, in agreement with the results of \cite{BV}.
 \end{remark}

We remind the reader that  measures satisfying the ergodicity axiom in addition to axioms OS0--3 (and hence relativistic quantum fields satisfying all Wightman axioms) can be obtained by an appropriate decomposition of the measure $\mu^{E,\rho}$, as explained in \cite[Section 19.7]{GJ}.

\appendix
\section{Measures on locally convex spaces}\label{meahb}
In this appendix we present some well-known results from measure theory on locally convex spaces; more detailed review may be found in \cite[Section 7.13]{B}. For background on topological vector spaces the reader is referred to \cite{T}. We use the notation introduced in the beginning of Section \ref{refposlimea} and by measure we always mean a {\em finite Radon measure}.

\subsection{Characteristic functionals}
Let $\mathcal{F}$ be a real locally convex vector space and let $\mathcal{F}\,'$ be its dual space considered with the weak* topology. The {\em characteristic functional} (or Fourier transform) of a measure $\mu$ on $\mathcal{F}\,'$ is a functional on $\mathcal{F}$ defined by
\begin{equation}\label{charde}
S_{\mu}(f)=\int_{\mathcal{F}\,'}e^{i\phi(f)}d\mu(\phi), \quad f \in \mathcal{F}.
\end{equation}
The functional $S_{\mu}$ is {\em positive definite}, i.e. for every $n$ and every $f_{1},\ldots, f_{n} \in \mathcal{F}$ the matrix whose $ij$-th entry is given by $S_{\mu}(f_{i}-f_{j})$ is positive semi-definite. Recall that a sequence of measures $\mu_{k}$  on a topological space $X$ {\em converges weakly} to a measure $\mu$ if
$$ \int_{X}fd\mu_{k} \rightarrow  \int_{X}fd\mu$$
for every bounded continuous function $f$ on $X$. Clearly the weak convergence of sequence of measures implies the convergence of their characteristic functionals. It turns out that the converse is also true for finite dimensional spaces (L\'{e}vy's theorem) and also for a wide class of locally convex spaces if equicontinuity is assumed.

We recall that  $\mathcal{D}(M)$ is a barrelled nuclear locally convex space for every smooth manifold $M$.

\begin{theorem}\label{levystheo} (cf. \cite[Chapter III, Corollary 2.6 and Example 2.3]{DF}) A sequence of measures on $\mathcal{D}'(M)$ converges weakly if their characteristic functionals are equicontinuous at zero and converge pointwise.
\end{theorem}

\subsection{Constructing measures from functionals}
The following result is known as the Bochner-Minlos theorem.
\begin{theorem}(see  \cite[Chapter III, Theorem 1.3]{DF})\label{minth} Let $\mathcal{F}$ be a nuclear locally convex space. For a functional on $\mathcal{F}$ to be the characteristic functional of a measure on  $\mathcal{F}\,'$ it is sufficient, and necessary if $\mathcal{F}$ is barrelled, that it be positive definite and continuous at zero.
\end{theorem}
We note that the argument proving Theorem \ref{minth} implies the following fact.
\begin{corollary}\label{unicountaddm} Let $\mu_{k}$ be a sequence of probability measures on a nuclear locally convex space $\mathcal{F}$. If the sequence of characteristic functionals $S_{\mu_{k}}$ is equicontinuous at zero the measures $\mu_{k}$ are uniformly countably additive in the sense that $\sum_{i=1}^{n}\mu_{k}(U_{i})$  converges to $\mu_{k}(\cup_{i=1}^{\infty}U_{i})$ uniformly in $k$ as $n \rightarrow \infty$ for every sequence $U_{i}$ of pairwise disjoint Borel subsets of $\mathcal{F}$.
\end{corollary}

We recall that a measure $\mu$ on a locally convex space $\mathcal{F}$ is called {\em Gaussian} if for every $\phi \in \mathcal{F}\,'$ the image measure $\phi^{\bullet}\mu$ on $\mathbb{R}$ is Gaussian.

\begin{example}\label{gauhilex}
Assume that the nuclear space $\mathcal{F}$ is continuously embedded into a real Hilbert space $\mathcal{H}$. Then by Theorem \ref{minth} every nonnegative selfadjoint bounded operator $C$ on $\mathcal{H}$ defines a mean zero Gaussian probability measure $\mu_{C}$  whose characteristic functional is given by
\begin{equation}\label{gauchar} S_{\mu_{C}}(f)=e^{-\frac{1}{2}\langle Cf, f \rangle}, \quad f \in \mathcal{F}.
\end{equation}
The operator $C$ is called the {\em covariance} of the measure $\mu_{C}$.
\end{example}

\begin{example}\label{gausmoex}
Similarly, a nonnegative selfadjoint operator $C$ on $L^{2}(M)$ that extends to a smoothing operator, i.e. an operator mapping continuously $\mathcal{D}'(M)$ into $\mathcal{D} (M)$, defines a unique centered Gaussian probability measure on $\mathcal{D}(M)$. Indeed, in this case the functional given by (\ref{gauchar}) is defined on $\mathcal{D}'(M)$ and we obtain via Theorem \ref{minth} a measure on $\mathcal{D}''(M)=\mathcal{D}(M)$.
\end{example}

\begin{remark}\label{remonincmes} In this article we shall use the same notation for a measure defined on $\mathcal{D}(M)$ and its image (pushforward) measure on $\mathcal{D}'(M)$ under the natural inclusion $\mathcal{D}(M) \hookrightarrow \mathcal{D}'(M)$. Taking this image clearly amounts to restricting the characteristic functional of the initial measure from $\mathcal{D}'(M)$ to $\mathcal{D}(M)$. In this situation we shall say that the image measure is supported in $\mathcal{D}(M)$.
\end{remark}

\section{Equicontinuity and uniform integrability}\label{limmesdistrsec}
In this appendix we use the notation introduced in the beginning of Section \ref{refposlimea} as well as the notation and results from Appendix \ref{meahb}.
\subsection{Equicontinuity and existence of limits}\label{equilimisubs}
In this subsection we present a simple sufficient condition for the existence of a limit point of certain sequences of probability measures on $\mathcal{D}'(M)$. The following generalization of the Arzel\`{a}-Ascoli theorem is a special case of \cite[Chapter X, \S 2, No. 5, Corollary 1]{Bu}.
\begin{theorem}\label{genaath} Let $\{\Phi_{n}\}_{n \in \mathbb{N}}$ be an equicontinuous sequence of complex-valued functions on a topological space $X$ and assume that $\Phi_{n}(x)$ is bounded for every $x \in X$. Then $\{\Phi_{n}\}$ contains a subsequence converging uniformly on compact subsets to a continuous function on $X$.
\end{theorem}

In what follows $\mathcal{V}$ stands for either $\mathcal{D}(M)$ or $\mathcal{D}'(M)$. Let $\mu_{C_{k}}, \, k=1,2,3,\ldots$ be a sequence of  centered Gaussian probability measures on $\mathcal{V}$ with covariance operators $C_{k}$ (cf. Theorem \ref{minth} and Eq. (\ref{gauchar})).
Let $\rho_{k}$ be sequence of densities on $\mathcal{V}$, i.e. $\rho_{k} \in L^{1}(\mu_{C_{k}})$ and $\rho_{k}\geq 0$. We define a sequence of probability measures on $\mathcal{V}$ (and hence on $\mathcal{D}'(M)$, cf. Remark \ref{remonincmes}) by setting
\begin{equation}\label{mukdeffinvo}\mu^{C,\rho}_{k}=N_{k}\rho_{k}\mu_{C_{k}},
\end{equation}
where
$$ N_{k}^{-1}=\int_{\mathcal{V}} \rho_{k}d\mu_{C_{k}}$$
are normalization constants.
Let $N$ be another Riemannian manifold and consider a sequence of operators $I_{k}: \mathcal{D}(N) \rightarrow \mathcal{D}(M)$.
\begin{proposition}\label{tightcrit} Assume that the operators $C_{k}$ and $I_{k}$ are uniformly $L^{2}$-bounded and that the densities $\rho_{k}$ satisfy
\begin{equation}\label{uestfornn}
\| \rho_{k}\|_{L^{1}(\mu_{C_{k}})}\geq  K_{1}, \quad  \| \rho_{k}\|_{L^{2}(\mu_{C_{k}})}\leq K_{2}
\end{equation}
for some positive constants $K_{1}$ and $K_{2}$.

(a) Then the sequence $S_{\mu^{C,\rho}_{k}}\circ I_{k}$ of characteristic functionals  on $\mathcal{D}(N)$ is equicontinuous.

(b) There exists a subsequence of $S^{C,\rho}_{\mu_{k}}\circ I_{k}$ converging to the characteristic functional $S_{\mu_{0}}$ of a probability measure $\mu_{0}$ on $\mathcal{D}'(N)$.
\end{proposition}
\begin{proof}
(a) It follows from the Cauchy-Schwartz inequality that
$$N_{k} \|\Psi\rho_{k}\|_{L^{1}(\mu_{C_{k}})} \leq N_{k} \|\Psi\|_{L^{2}(\mu_{C_{k}})}\|\rho_{k}\|_{L^{2}(\mu_{C_{k}})}$$
for every bounded continuous functional $\Psi$ on $\mathcal{D}'(M)$. Hence, using the bounds (\ref{uestfornn}), one obtains
\begin{equation}\label{lonetoltwoe}
\|\Psi\|_{L^{1}(\mu^{C,\rho}_{k})} \leq K^{-1}_{1}K_{2}\|\Psi\|_{L^{2}(\mu_{C_{k}})}.
\end{equation}

Now a short computation employing (\ref{gauchar}) shows that
\begin{equation}\label{gaudiff}\int_{\mathcal{V}}{\bigl|e^{i\phi(f)}-  e^{i\phi(g)} \bigr|^{2}d\mu_{C_{k}}(\phi)}=
2(1-e^{-\|C^{1/2}_{k}(f-g)\|_{L^2}^{2}/2}),
\end{equation}
where $f,g \in \mathcal{D}(M)$ and $\phi(f), \phi(g)$ are the corresponding functionals on $\mathcal{D}'(M)$. Combining
(\ref{lonetoltwoe}) and (\ref{gaudiff}), one finds that
$$|S_{\mu^{C,\rho}_{k}}(I_{k}f)- S_{\mu^{C,\rho}_{k}}(I_{k}g)|\leq \text{const}\cdot (1-e^{-\|C^{1/2}_{k}I_{k}(f-g)\|_{L^2}^{2}/2}), \quad f,g \in \mathcal{D}(N),$$
from which the equicontinuity of $S_{\mu^{C,\rho}_{k}}\circ I_{k}$ follows since $C_{k}$ and $I_{k}$ are uniformly bounded and  convergence in $\mathcal{D}(N)$ implies convergence in $L^{2}(N)$.

(b) The existence of a converging subsequence of $S_{\mu^{C,\rho}_{k}}\circ I_{k}$ follows from part (a) and Theorem \ref{genaath}. The limit of such a subsequence is the characteristic functional of a probability measure  by Theorem \ref{minth}.

\end{proof}

\subsection{Equicontinuity and analyticity}\label{compfanalysubsec}
Recall that a functional on a complex locally convex space is called {\em (Fr\'{e}chet) analytic} if it is continuous and its restrictions to finite dimensional subspaces are analytic (see e.g. \cite[Chapter 3]{H})

\begin{theorem}\label{locbouequith} (cf. \cite[Theorem 3.1.5(c)]{H}) A uniformly locally bounded sequence of analytic functionals on a  complex locally convex space  is equicontinuous.
\end{theorem}

In what follows, we denote the complexification of a real vector space $\mathcal{F}$  by $\mathcal{F}_{\mathbb{C}}$. If $\mu$ is a measure on the dual of a real locally convex space $\mathcal{F}$ we denote the extension of its characteristic functional to $\mathcal{F}_{\mathbb{C}}$ by $S^{\mathbb{C}}_{\mu}$. This extension (which may or may not exist) is defined by setting
\begin{equation}\label{compx}\phi(f_{1}+if_{2})=\phi(f_{1})+i\phi(f_{2}), \quad  f_{1},f_{2} \in \mathcal{F}
\end{equation}
in the integral in (\ref{charde}).

If $\mu_{C}$ is a Gaussian measure on  $\mathcal{D}'(M)$ with covariance operator $C$, the extension $S^{\mathbb{C}}_{\mu_{C}}$ is analytic. Extending $C$ to $\mathcal{D}(M)_{\mathbb{C}}$, one sees that the right-hand side of (\ref{gauchar}) extends to an analytic functional on $\mathcal{D}(M)_{\mathbb{C}}$ as well. It follows that (\ref{gauchar}) continues to hold for these two extended functionals and hence one has

\begin{equation}\label{extgauseti}
|S^{\mathbb{C}}_{\mu_{C}}(f)|\leq e^{K\|C^{1/2}\|_{L^{2}}^{2}\|f\|^{2}_{L^2}}  ,\quad f \in \mathcal{D}(M)_{\mathbb{C}},
\end{equation}
where $K$ is a positive constant independent of $C$.

\subsection{On uniform integrability}\label{uniinteggrasubs}
The classical notion of uniform integrability has been generalized to sequences of measures in \cite{Se}. Let $\mu_{k}$ be a sequence of probability measures on a topological space $X$ and let $\Phi_{k}$ be a sequence of Borel measurable functions on $X$. We say that $\Phi_{k}$ is {\em uniformly $\mu_{k}$-integrable} if $\sup_{k}\int |\Phi_{k}|d \mu_{k} < \infty$ and for every $\varepsilon >0$ there exists $\delta >0$ such that for every sequence $X_{k}$ of Borel subsets of $X$ with $\sup_{k}\mu_{k}(X_{k}) < \delta $ one has $\sup_{k}\int_{X_{k}} |\Phi_{k}|d \mu_{k} < \varepsilon$. (This is one of the equivalent definitions of uniform integrability stated in \cite[Lemma 2.5]{Se}.)

We say $\Phi_{k}$ converges to a Borel measurable function $\Phi$ {\em in  $\mu_{k}$-measure} if
$$\mu_{k}\{x \in  X: | \Phi_{k}(x)-\Phi(x)|>\varepsilon \} \rightarrow 0$$ as $k \rightarrow \infty$ for every $\varepsilon >0$.
The following generalization of Vitali's integral convergence theorem to sequences of measures is a special case of \cite[Theorem 2.7]{Se}.

\begin{theorem}\label{genvitalithe} Let $ 1 \leq p < \infty$. Suppose that $\int_{X} |\Phi_{k}|^{p}d \mu_{k} < \infty$ and that $|\Phi|^{p}$ is uniformly $\mu_{k}$-integrable. Then the following statements are equivalent.

(1) $ \, \lim_{k \rightarrow \infty} \int_{X} |\Phi_{k}-\Phi|^{p}d \mu_{k}=0$.

(2) $\Phi_{k}$ converges to $\Phi$ in $\mu_{k}$-measure and $\Phi_{k}$ is uniformly $\mu_{k}$-integrable.
\end{theorem}

The following lemmas will be useful in establishing the convergence in measure and uniform integrability conditions appearing in Theorem \ref{genvitalithe}.

\begin{lemma}\label{gendomconvspcase} Let $\mu_{k}$ be a sequence of probability measures on a nuclear locally convex space $\mathcal{F}$ with equicontinous sequence of characteristic functionals. Then $\Phi_{k}$ converging to $\Phi$ pointwise implies that $\Phi_{k}$ converges to $\Phi$ in $\mu_{k}$-measure.
\end{lemma}
\begin{proof} By  Corollary \ref{unicountaddm} the measures $\mu_{k}$ are uniformly countably additive, hence the standard argument establishing that on a finite measure space pointwise convergence implies convergence in measure (see e.g. \cite[Theorem 2.2.3]{B}) implies that $\Phi_{k}$ converges to $\Phi$ in $\mu_{k}$-measure.
\end{proof}

\begin{lemma}\label{pgroneunile} Suppose that there exists $p>1$ such that
$\int_{X}  |\Phi_{k}|^{p}d \mu_{k} < K$
for all $k$ and some positive constant $K$. Then the sequence $\Phi_{k}$ is uniformly $\mu_{k}$-integrable.
\end{lemma}
\begin{proof} By H\"older's inequality one has for every Borel set $U \subset X$
\begin{equation}\label{hointulem}\int_{U} |\Phi_{k}|d \mu_{k}\leq \mu_{k}(U)^{1/q}\left(\int_{X} | \Phi_{k}|^{p}d \mu_{k}\right)^{1/p},
\end{equation}
whenever $1/p+1/q=1$. Thus given $\varepsilon >0$ we can choose $\delta=\varepsilon^{q} K^{-\frac{q}{p}}$ so that $\sup_{k}\mu_{k}(X_{k}) < \delta $ for any sequence $X_{k}$ of Borel subsets of $X$ implies by virtue of (\ref{hointulem}) that $\sup_{k}\int_{X_{k}} |\Phi_{k}|d \mu_{k} < \varepsilon$.
\end{proof}


\end{document}